\documentclass[journalpaper, 9pt]{article}
\usepackage[paper=letterpaper,centering,margin=1.0in]{geometry}

\usepackage[american]{babel}
\usepackage{epsfig}
\usepackage{epstopdf}
\usepackage{amsthm,amsfonts,amsmath,amssymb,amstext,latexsym}
\usepackage{hyperref}

\usepackage{color}
\allowdisplaybreaks

\newtheorem{theorem}{Theorem}

\newtheorem{conjecture}{Conjecture}
\newtheorem{remark}{Remark}

\newtheorem{proposition}{Proposition}

\newtheorem{definition}{Definition}

\newtheorem{assumption}{Assumption}

\newcommand{\Q}{\overline{Q}}
\newcommand{\UL}{\overline{U}}
\newcommand{\T}{\overline{T}}
\newcommand{\B}{\overline{B}}
\newcommand{\I}{\overline{I}}
\newcommand{\V}{\overline{V}}

\newcommand{\bareta}{\bar{\eta}}

\newcommand{\E}{\mathbb{E}}
\newcommand{\pp}{\mathbb{P}}

\newcommand{\bydef}{:=}

\title{Stability and Optimization of Speculative Queueing Networks}

\date{}
\author{
J. Anselmi$^1$ and N. Walton$^2$\\
{\small $^1$: Univ. Grenoble Alpes, CNRS, Inria, Grenoble INP, LIG, 38000 Grenoble, France.}\\
{\small $^2$: University of Manchester, UK }\\
jonatha.anselmi@inria.fr, ~neil.walton@manchester.ac.uk
}

\begin{document}
\maketitle

\begin{abstract}
We provide a queueing-theoretic framework for job replication schemes based on the principle ``\emph{replicate a job as soon as the system detects it as a \emph{straggler}}''. This is called job \emph{speculation}. Recent works have analyzed {replication} on arrival, which we refer to as \emph{replication}. Replication is motivated by its implementation in Google's BigTable. However, systems such as Apache Spark and Hadoop MapReduce implement speculative job execution. The performance and optimization of speculative job execution is not well understood.
To this end, we propose a queueing network model for load balancing where each server can speculate on the execution time of a job. 
Specifically, each job is initially assigned to a single server by a frontend dispatcher. 
Then, when its execution begins, the server sets a timeout. 
If the job completes before the timeout, it leaves the network, otherwise the job is terminated and relaunched or resumed at another server where it will complete. 
We provide a necessary and sufficient condition for the stability of speculative queueing networks with heterogeneous servers, general job sizes and scheduling disciplines. We find that speculation can increase the stability region of the network when compared with standard load balancing models and replication schemes. We provide general conditions under which timeouts increase the size of the stability region and derive a formula for the optimal speculation time, i.e., the timeout that minimizes the load induced through speculation. We compare speculation with redundant-$d$ and redundant-to-idle-queue-$d$ rules under an $S\& X$ model. For light loaded systems, redundancy schemes provide better response times. However, for moderate to heavy loadings, redundancy schemes can lose capacity and have markedly worse response times when compared with the proposed speculative scheme.
\end{abstract}

\section{Introduction}

Jobs taking longer than expected to complete, so-called \emph{stragglers}, severely impact the performance of computer systems. 
Initial work on MapReduce note that stragglers have a significant impact on performance and occur for reasons that do not necessarily depend on the magnitude of the job's computational requirements but instead on the configuration and runtime of the underlying platform \cite{dean2008mapreduce}. 
Causes for straggling jobs include run-time contention phenomena among CPU cores, processor caches, memory bandwidth, network bandwidth~\cite{xu13,USENIX12}; 
bugs and errors in code \cite{dean2008mapreduce};
unfortunate disk seek times or head locations~\cite{Hai2005}; or, energy requirements, power and temperature constraints, and maintenance activities \cite{Google}. 
All of which can cause a job to straggle significantly beyond its inherent computational requirement.

To mitigate the effect of stragglers on system-wide performance, researchers are currently proposing redundant, or {hedged}, requests \cite{Ananthanarayanan10,Google,Hai2005,Vulimiri2013,hopper,Ananthanarayanan13}.
In this respect, there are two underlying principles: 
either replicate 
\begin{quotation}
``\emph{replicate a job upon its arrival and use the results from whichever replica responds first}'',~\cite{Google};
\end{quotation}
or speculate
\begin{quotation}
``\emph{replicate a job as soon as the system detects it as a straggler}'', \cite{Ananthanarayanan10}.
\end{quotation}

The former is referred to as job replication, and here all redundant replicas are usually canceled as soon as one completes or starts service. The latter is referred to as job speculation~\cite{zaharia2008improving}, and here long jobs are usually killed after some time and replicated elsewhere. 
Clearly, the straggler detection monitor, or \emph{timeout} rule, needs to be smart enough to distinguish between jobs whose intrinsic size is large and jobs that are taking longer than expected to complete because of some unfortunate runtime phenomenon.
There is increased interest in replication as it is employed by Google's BigTable, while the speculation principle is used in Apache Spark and Hadoop MapReduce.\footnote{For instance, Apache Spark applies the attribute \texttt{spark.speculation} while Hadoop MapReduce uses task configuration \texttt{mapred.map.tasks.speculative.execution}.} For instance, it was reported in \cite{hopper} that speculative tasks account for 25\% of all tasks in Facebook's Hadoop cluster.

Both principles provide effective mitigation techniques for the straggler problem but which one is preferable is not currently clear. Replication may lead to significant additional resource costs while speculation may increase latency.
One advantage of speculation is that only the long jobs are replicated, though this may come at the cost of an increased latency because a cloned job is executed after the system detected it as a straggler. 
Our objective is to provide the first theoretical framework for speculation in a stochastic and dynamic setting. We then use this to understand the impact of speculation and how to optimize its performance.

Recently, a number of theoretical works appeared in the literature for evaluating the performance of replication-based schemes in a stochastic and dynamic setting; see Section~\ref{literature} for a review.
These consider systems that replicate on arrival.
Speculation, however, has received far less attention.

\subsection{Model, Results and Contributions}

In this paper, we develop the first theoretical framework for evaluating the delay performance achieved by the speculation principle in  a queueing network.
We consider a stochastic system of parallel queues, each with its own server. Upon arrival, each job is dispatched to a random queue and, when a job's service is first initiated, a timeout is set. If the timeout is reached and the job is not completed, then the job is killed and re-routed to another randomly selected queue where it must be completed. 


Our first result, Theorem~\ref{th1}, characterizes the \emph{stability} region induced by speculation.
We show that speculative queueing networks are stable (positive Harris recurrent\footnote{We point the reader unfamiliar with positive Harris recurrence to \cite[Section~3]{dai1995positive}}) under the usual stability condition: that is, 
the nominal load at each queue is less than one.
This result is non-trivial because: i) our network includes feedback mechanisms with mixture of different job size distributions; ii) we consider a general class of work-conserving scheduling disciplines and service time distributions; and iii) the stability region found differs from the one achieved by standard load balancing schemes such as join-the-shortest-queue. It is well known that under i) and ii) the ``usual" stability condition may not be sufficient for positive recurrence~\cite{bramson1994,328805,gamarnik2005}.

Since timeouts induce a controllable constraint on the processing time of each job, we then consider the problem of optimal timeout design. 
First, we note that speculation can increase the system load and thus decrease the stability region, because part of the job is effectively served twice. Therefore, in Theorem \ref{th3}, we derive criteria for speculation to increase the stability region and show that these criteria are satisfied in realistic scenarios.
Then, we investigate optimal design and characterize the timeout that maximizes the size of the stability region in terms of a non-linear one-dimensional equation given in Theorem~\ref{th_optimal_timeout}. This optimal timeout can be easily solved numerically and in some cases also analytically.

Overall, we use our results to investigate the advantage of speculation by performing two comparisons:
\begin{itemize}
\item  First, we compare speculative and \emph{standard} load balancing in the stationary regime; ``standard'' means no replication and no speculation, though dynamic information between servers and dispatcher(s) can be exchanged, as in power-of-$d$. 
We establish when it is advantageous to introduce timeouts as means to improve performance in an existing standard load balancing system. 
An immediate consequence of Theorem~\ref{th1} is that a timeout increases the stability region if and only if the expected remaining time of the current execution is greater than the conditional expected service time at the second queue. 
We find large sets of timeouts having this property for important classes of service time distributions; e.g., Pareto, hyper-exponential and bimodal, see Section~\ref{sec:comp_bernoulli}.
We can then refine this analysis to give the optimal timeout for a given distribution in Theorem~\ref{th_optimal_timeout}.

\item 
Second, we compare speculation and replication in the stationary regime.
This comparison is difficult to perform analytically because a satisfactory characterization of the stability region induced by replication is currently unknown even when considering homogeneous servers, outside special cases; see \cite{Borst2020,MENDELSON2021113,Raaijmakers2019,anton2019stability} and Section~\ref{literature} for further details. However, we can provide reasoning and scenarios where speculation will improve response times when compared with replication. In particular, replication is bad when the shortest replica is stuck in queue, and instead service is completed by serving one or more larger replicas. 
The reasoning bears out in numerical simulations.
Although replication provides better average response times in a light-load regime, which is to be expected, a large number of simulations indicate that speculative load balancing with the optimal timeout yields an increased stability region and short response times over a wide range of loads.
\end{itemize}

Our work is the first developing either of the above comparisons.
Concisely, in terms of stability we find that speculative load balancing can be more advantageous than standard load balancing and replication.

\subsection{Organization}

The remainder of the paper is organized as follows.
In Section~\ref{literature}, we review the relevant literature.
In Section~\ref{sec:model}, we introduce speculative queueing networks. In Section \ref{sec:stability}, we present our stability theorem. 
In Section~\ref{sec:timeout_design}, we address the problem of designing optimal timeouts and compare the stability region achieved within our approach with standard load balancing and existing replication schemes.
In Section~\ref{sec:Large_Response}, we present a conjecture about the mean waiting time obtained when servers are homogeneous.
Finally, Section~\ref{sec:conclusions} draws the conclusions of this work and outlines further research.

\section{Literature Review}
\label{literature}
Over the last decade there has been interest in the performance analysis of policies that dispatch to the shortest among a set of queues, for instance, power-of-two-choices and join-the-idle-queue; see~\cite{van2018scalable} for a survey.
However, more recently there has been an increased focus on job replication and in the following, we review existing replication approaches; for further details, we refer the reader to the recent survey~\cite{gardner2020product}.

Within replication, multiple replicas of a job are placed at different servers and redundant replicas are canceled, either when the first replica completes service (cancel-on-completion) or when the first replica initiates service (cancel-on-start). 
It is striking that both cancel-on-start and cancel-on-complete redundancy models are analytically tractable under independent exponentially distributed job sizes. This is first shown for cancel-on-complete in \cite{gardner2016queueing}, where the reversibility results of the order-independent queue\cite{berezner1995quasi} are applied. 
The stationary distribution for cancel-on-start is found in \cite{ayesta2019redundancy}, which applies results from~\cite{visschers2012product}.
The paper \cite{joshi2017efficient} is the first to analyze the latency and computing cost of the cancel-on-start and cancel-on-finish redundancy strategies.
Subsequent works expand these results and thus our understanding of redundancy models~\cite{gardner2017redundancy}.
The stationary distributions found are fully explicit; however, it is often hard to derive tractable formulas for key performance metrics. 
For this reason, existing works develop diffusion limits and mean field approximations for these models \cite{cardinaels2020redundancy,hellemans2020heavy}.
Under i.i.d. exponential replica sizes, it is noted that there is neither loss nor gain of capacity due to redundancy job replication \cite{bonald2017balanced}. However, it is also noted in \cite{bettermodelTON} that the independent exponential distributed job sizes is not a reasonable assumption, as job sizes will be correlated and slow service tends to result in heavy tailed service times. 
For this reason, works investigate the impact of replication under a variety of job size distributions and service disciplines~\cite{hellemans2019performance}. The results in~\cite{Borst2020,anton2019stability} find that in general capacity is lost due to increased utilization of redundant jobs, though a full understanding of the stability region induced by cancel-on-completion schemes is currently unclear; see also~\cite{MENDELSON2021113}. 

Another similar approach is delayed replication. Here, each job is sent to an arbitrary server and if the server has not finished service of this job within some fixed time that includes both waiting and service, the job is replicated elsewhere.
From the point of view of a single queue, this approach is somewhat equivalent to a queueing system with abandonments or reneging~\cite{Haight1959}. Delayed replication is investigated in~\cite{delayedrep} for homogeneous servers and under the assumption of ``asymptotic independence'', though the structure of the resulting stability region is not explicited.

While several recent works investigate replication, studies of speculation are largely empirical and analyze the performance and dependencies of specific platform implementations~\cite{Ananthanarayanan10,Google,Hai2005,Vulimiri2013,hopper,Ananthanarayanan13,hopper}. 
%
An early analysis of timeouts in queueing is given by \cite{TAGS2002}, where all jobs initially join a single queue and then jobs that timeout are sequentially sent along a line of further queues.
Some existing theoretical works focus on static settings with a finite number of jobs and no queueing~\cite{8884664,Joshi19,101145}.

The queueing and stability behavior of job speculation is not well-understood and we investigate its impact.
We investigate stability using the fluid limit approach of Dai~\cite{dai1995positive,bramson2008stability}. We note that job sizes at different queues change due to both timeouts and correlations between jobs. In general, such networks can exhibit non-standard stability behavior; see, for instance, Bramson \cite{bramson1994}. Although stability is non-trivial for these networks, we find that stability still holds when the nominal load is below one. This is important as the nominal load can be reduced through the use of speculation. 

Further, we characterize the optimal timeout for a speculative load-balancing network. 
Again there are few works that investigate optimal timeouts, though we note that the paper~\cite{xu2016optimization} provide upper and lower bounds for a constrained utility optimization. This is then approximated for a timeout optimization under independent task-based service time distributions and in the presence of idle servers. In contrast, we consider a queueing framework under an S\&X model where the optimal timeout minimizes load and thus maximizes stability. The use of Markov decision processes and optimal stopping theory~\cite{peskir2006optimal} analysis appears to be somewhat new to this line of literature.
We note the contemporaneous work of \cite{joshi2020synergy} consider a discrete time MDP formulation of optimal replication. This approach (along  with associated implications for reinforcement learning and load balancing) would seem an important area of future investigation.

%


\section{Speculative Queueing Network}
\label{sec:model}

We consider a network composed of~$N$ queues working in parallel. Queues have an infinite buffer, are work-conserving and are heterogeneous, in the sense that they process jobs at possibly different rates. Specifically, queue~$i$ processes jobs with rate $\mu_i>0$, $i=1,\ldots,N$, and without loss of generality we assume that $\sum_{i=1}^N \mu_i=d$.
Jobs join the network via an exogenous Poisson process with rate $\lambda N$ and are initially dispatched to queue~$i$ independently with probability $p_{0,i}$.
When a job starts its execution on its designated queue, say $i$, a clock is initialized and a \emph{timeout}, $\tau_i$, is set. 
If the timeout is reached and the job has not completed, then it is killed and re-routed independently to another queue where it is either resumed or re-executed from scratch.
At this second queue, the job is then served until complete and leaves the network afterwards.

\subsection{Dynamics and Assumptions}
\label{sec:notation}

Let $\tau\bydef(\tau_1,\ldots,\tau_N)\in\mathbb{R}_+^d$ denote the timeout vector.
The $n$-th job that joins the network has associated the random variables:
$\xi(n)\in\mathbb{R}_+$, the interarrival time between the $n$-th and $(n-1)$-th jobs; 
$\eta(n)\bydef(\eta_1(n),\eta_2(n))$, the service requirements of job $n$ where $\eta_j(n)\in\mathbb{R}_+$ represents the service time associated to its $j$-th execution if processed at unit rate;
$\ell(n)\bydef(\ell_1(n),\ell_2(n))$, where $\ell_j(n)\in\{1,\ldots,N\}$ denotes the queue associated to the $j$-th visit performed by job $n$; and
$\tau(n)\bydef \tau_{\ell_1(n)}$, the timeout of job $n$.

On the event $$\frac{\eta_1(n)}{\mu_{\ell_1(n)}}\le \tau(n),$$
job $n$ completes service at $\ell_1(n)$ before the timeout $\tau(n)$ is reached and it leaves the network. In this case, the random variables $\eta_2(n)$ and $\ell_2(n)$ are not used. 
On the complementary event, job $n$ is re-routed to queue $\ell_2(n)$ where it receives service for $\frac{\eta_2(n)}{\mu_{\ell_2(n)}}$ time units and leaves the network afterwards.

Within this notation, the primitive sequences driving the dynamics of the stochastic network under investigation are 
$\{\xi (n), n\ge 1\}$, 
$\{\eta(n), n\ge 1\}$ and
$\{\ell(n), n\ge 1\}$, all defined on a common probability space.
We assume that these sequences are i.i.d. and mutually independent. 
We note however that, for fixed $n$, $\eta_1(n)$ and $\eta_2(n)$ may have arbitrary dependency. For simplicity of notation, we will refer to $\eta_1(1)$ and $\eta_2(1)$ simply as $\eta_1$ and $\eta_2$, respectively.
Because ~$\eta_1$ and~$\eta_2$ may be dependent and have a possibly different distribution, this allows us to models two different scenarios:
\begin{itemize}
 \item $\eta_1$ and $\eta_2$ are equal in distribution: upon re-routing, the job is killed and then re-started from scratch
 \item $\eta_2$ is stochastically smaller than $\eta_1$: upon re-routing, the job is stopped and then resumed.
\end{itemize}
The case of identical replicas, i.e., $\eta_1=\eta_2$ ($\omega$-per-$\omega$) is permitted but it should be clear that in this case speculation does not help because the amount of work to be done does not change after re-routing. 
We remark that speculation is motivated by the fact that something can go wrong on the servers' runtimes, i.e., a given job may have different service times if executed on different machines. This rules out the case of identical replicas. 
We also require that $\E[\eta_1(n)+\eta_2(n)]<\infty$  and that $\pp(\frac{\eta_1}{\mu_i}\le  \tau_i)>0$ for all $i$, which means that jobs complete and have a chance of completing before the timeout.


We assume that each $\xi(n)$ is exponentially distributed with mean $(\lambda N)^{-1}$, so that arrivals occur according to a Poisson process of rate $\lambda N$. 
For all $i,j\in\{1,\ldots,N\}$, let $p_{i,j}\bydef \pp(\ell_2(n)=j \mid \ell_1(n)=i  )$.
We assume that $p_{i,j}=p_{1,j}$ for all $i$, i.e., the internal routing probabilities depend on the destination~$j$ but not on the source~$i$. This assumption makes all jobs `homogeneous in terms of re-routing' and is required in the proof of Theorem~\ref{th3}; without this assumption, the proposed queueing network may be unstable even under the usual stability condition, as in~\cite{bramson1994}. 



Since queues speculate on service times, we refer to the proposed model as a \emph{speculative queueing network}.

\subsection{Scheduling Disciplines}
\label{sec:SD}

Each queue operates under a work-conserving head-of-the-line scheduling discipline. We recall that a service discipline is head-of-the-line if within each class at the queue jobs are served in order their arrival \cite{bramson2008stability}. 
This includes First-come first-served (FCFS), Head-of-the-line processor-sharing (HLPS) or class-based priority disciplines.
On the other hand, Processor-Sharing (PS) is excluded, though we believe that the results presented in this paper apply by a separate argument \cite[Section 3.3]{Kelly79}; due to page constraints, we do not discuss this.
In our framework, the class of a job refers to the queue that the job is at and whether the job has been timed-out or not; see Section~\ref{sec:multiclass_rep} for further details.
With a class-based priority discipline, jobs are ranked according to their class. The higher the rank, the higher the priority. Depending on whether the discipline is preemptive or not, the job in execution may be stopped if a job with higher priority arrives.
For instance, within a given queue, a job that has been re-routed may be given higher (or lower) priority than job that have currently visited only one queue.

We observe that some scheduling disciplines may process more than one job at a time.
In this case, the clock associated to the execution of each job does not increase linearly over time, but in relation to the processing capacity devoted to the execution of the job itself.

\subsection{Further Notation}

We use the convention that products (resp. sums) over empty sets are one (resp. zero). 
The set of non-negative real numbers is denoted by $\mathbb{R}_+$.
The indicator function of $A$ is denoted by $\mathbb{I}_A$.
We use $|\cdot|$ to denote set cardinality
and $\|\cdot\|$ to denote the $L_1$ norm.
We let $p_0\bydef (p_{0,i}: i=1,...,N)$ and $p_1\bydef (p_{1,i}: i=1,...,N)$ denote the routing probability vectors upon job arrival in the network and re-routing, respectively.
Unless specified otherwise, indices $i$ and $j$ will range over~$\{1,\ldots,N\}$.
For $a,b\in\mathbb{R}$, $(a)^+=\max\{a,0\}$, $a\wedge b = \min\{ a, b\}$ and $a\vee b = \max\{a,b\}$.
For a differentiable function $t\mapsto f(t)$, $f'(t)$ denotes its derivative in~$t$.

\section{Stability}\label{sec:stability}

In this section, we first provide a multiclass representation of the speculative queueing network introduced in previous section and show that this representation belongs to the family of queueing networks investigated in~\cite{dai1995positive}. 
This will allow us to use Theorem 4.2 of~\cite{dai1995positive}, which provides a criterion to establish the stability (positive Harris recurrence) of the Markov process, say~$X$, describing the dynamics of the network under investigation. This criterion is expressed in terms of an associated \emph{fluid model}, which we define below.

\subsection{Multiclass Representation}
\label{sec:multiclass_rep}

We now provide a multiclass representation of our queueing network described in Section~\ref{sec:notation} where classes are used to distinguish between jobs that have timeout or not. The advantage of this is that we can directly represent jobs within the framework of Dai~\cite{dai1995positive} and Bramson~\cite{bramson2008stability}. Thus we can analyze our fluid model as a fluid limit as expressed within that framework. 



Let us consider the set of classes
$\mathcal{K}\bydef \mathcal{K}_1 \,\cup \mathcal{K}_2$ 
where 
$\mathcal{K}_1\bydef \{(k,i): k\in\{c,u\}, 1\le i \le N\}$ and $\mathcal{K}_2\bydef \{(i,j): 1\le i, j\le N\}$ represent the set of exogenous and endogenous classes, respectively. 
Further, we use `$c$' for jobs that will complete service at their first queue, and~`$u$' for jobs that will receive uncompleted service at their first queue (and thus will timeout).
On the same probability space used to define the speculative queueing network above,
we now construct a queueing network composed of such classes.

Jobs enter the network through class $(c,i)\in \mathcal{K}_1$ or $(u,i)\in \mathcal{K}_1$ for $i=1,...,N$. These jobs join queue~$i$ and, after processing at~$i$, a class-$(c,i)$ job leaves the network, while a class-$(u,i)$ job is re-routed and joins queue $j$ as a class $(i,j)\in \mathcal K_2$ job with probability~$p_{i,j}$. 

The inter-arrival times of jobs in class $(c,i)\in \mathcal{K}_1$ follow a Poisson process that is obtained by thinning
the Poisson process associated to $\{\xi (n), n\ge 1\}$ with respect to the Bernoulli process~$\{ B_n^{c,i},n\ge 1 \}$ where
\begin{align*}
B_n^{c,i}\bydef \mathbb{I}_{\left\{\ell_1(n)=i,\,\tfrac{\eta_1(n)}{\mu_i}\le \tau_i\right\}} \, .
\end{align*}
Under the Poisson thinning (aka splitting) property, the set of service times of class-$(c,i)$ jobs can be equivalently obtained by sampling independently from the distribution of~$(\frac{\eta_1}{\mu_i} | \frac{\eta_1}{\mu_i} \leq \tau_i)$. 
Similarly, the interarrival times of jobs in class $(u,i)\in \mathcal{K}_1$ follow a Poisson process that is obtained by thinning the Poisson process associated to $\{\xi (n), n\ge 1\}$ with respect to the Bernoulli process $\{ B_n^{u,i},n\ge 1 \}$
where
\begin{align*}
B_n^{u,i}\bydef \mathbb{I}_{\left\{\ell_1(n)=i,\,\tfrac{\eta_1(n)}{\mu_i}> \tau_i\right\}}
\end{align*}
and they have deterministic service times equal to $\tau_i$.

Under the Poisson thinning property, the set of service times of class-$(i,j)$ jobs can be equivalently obtained by sampling independently from the distribution of~$(\frac{\eta_2}{\mu_j} | \frac{\eta_1}{\mu_i} > \tau_i)$. 


The dynamics (routing decisions, arrival and service times) of the multiclass network above are equivalent to the dynamics of the speculative queueing network introduced in Section~\ref{sec:model}. Equivalences of this type are commonly applied in the analysis of queueing networks; for instance, see Section 2.7 of \cite{dai_harrison_2020}.


\subsection{Markov and Fluid Model}

We now define a continuous-time Markov process $X(t)$ that describes the dynamics of the multiclass queueing network described above.
%
%
%
%
Specifically, 
taken to be right continuous, consider
the pair
\begin{equation}
\label{QUV}
X(t)\bydef (\mathbb{Q}(t), V(t)).
\end{equation}
Here,
$\mathbb{Q}(t)=(\mathbb{Q}_1(t),\ldots,\mathbb{Q}_N(t))$ and $$\mathbb{Q}_i(t)=(k_{i,1},k_{i,2},...,k_{i,Q_i(t)})$$ 
where $k_{i,n} \in \mathcal K$ gives the class of the $n$-th job in queue $i$ and $Q_i(t)$ is the total number of jobs at queue~$i$. This captures how jobs are lined up in queue~$i$.
Also 
$V(t)=(V_k(t):k\in\mathcal{K})$ with $V_k(t)$ denoting the remaining service time of the class-$k$ job in execution and with the convention that $V_k(t)=0$ if such job does not exist.

The process $X=\{X(t), t\ge 0\}$, living on state space $\mathcal{X}$, is a piecewise deterministic Markov process (PDMP) and satisfies the strong Markov property; see page 362 in~\cite{Davis}.
If $X$ is positive Harris recurrent, we recall that a unique ergodic stationary distribution $\pi$ exists and
\begin{equation}
\lim_{t\to\infty} \frac{1}{t}\int_0^t Q_i(s)\, {\rm d}s = \int_{\mathcal{X}} q_i(x) \, \pi({\rm d}x)
\end{equation}
almost surely, for any initial configuration, where $Q_i(t)$ is the queue length of server $i$ at
time $t$ and $q_i(x)$ is the queue length of server $i$ for network state $x\in\mathcal{X}$.

Now, we define a fluid model for the dynamics of the queueing network under investigation. 
For $t\in\mathbb{R}_+$, let 
$\Q(t)=(\Q_k(t),k\in\mathcal{K})$, $\T(t)=(\T_k(t),k\in\mathcal{K})$ and $\I(t)=(\I_i(t),i=1,\ldots,N)$
and consider the following conditions
\begin{align}
%
\label{eq:FS1} &\Q_{c,i}(t) = \Q_{c,i}(0)  + \lambda d p_{0,i} \pp(\eta_1\le \tau_i\mu_i)\, t   - \frac{( \T_{c,i}(t) - \V_{c,i} )^+}{\E[\frac{\eta_1}{\mu_i}\mid \frac{\eta_1}{\mu_i}\le \tau_i]}  \ge 0\\
\label{eq:Q_ui}
&\Q_{u,i}(t) = \Q_{u,i}(0)  + \lambda d p_{0,i} \pp(\eta_1> \tau_i\mu_i)\, t  - \frac{( \T_{u,i}(t) - \V_{u,i} )^+}{\tau_i}   \ge 0 \\
\label{eq:FS2}
&\Q_{j,i}(t) = \Q_{j,i}(0)  - \frac{( \T_{j,i}(t) - \V_{j,i} )^+}{\E[\frac{\eta_2}{\mu_i}\mid\frac{\eta_1}{\mu_j}>\tau_j]}  + \frac{ p_{j,i}}{\tau_j}( \T_{u,j}(t) - \V_{u,j} )^+  \ge 0,\,\quad\forall (j,i)\in\mathcal{K}_2\\
\label{eq:Tover}
&\T_k(t) \mbox{ is nondecreasing and starts from zero} ,\quad \forall k\in\mathcal{K}\\
\label{eq:Iover}
&\I_i(t) = t -  T_{c,i}(t)- \bar T_{u,i}(t) - \sum_{j\neq i} \T_{j,i}(t) \mbox{ is nondecreasing} \\
\label{eq:FSL}
&\int_0^\infty \bigg( \Q_{c,i}(t) +\Q_{u,i}(t) + \sum_{j\neq i} \Q_{j,i}(t) \bigg)\,  {\rm d} \I_i(t) =0 
\end{align}
for all $i \in\{1,\ldots,N\}$
where $(x)^+=\max(x,0)$ and $\overline{V} \in\mathbb{R}_+^{|\mathcal{K}|}$ is interpreted as the remaining service time vector.
These conditions define our \emph{fluid solutions}. A \emph{fluid solution} is any solution $(\Q(\cdot),\T(\cdot))$ to \eqref{eq:FS1}-\eqref{eq:FSL}.
The quantities in \eqref{eq:FS1}-\eqref{eq:Iover} have the following interpretation:
at the fluid scale, $\Q_k(t)$ is the amount of jobs of class $k$ at time $t$,
$\T_k(t)$ is the cumulative time dedicated to the processing of class $k$ jobs by time $t$, and
$\I_i(t)$ is the cumulative idle time of queue~$i$ by time~$t$.
For instance, Equation~\eqref{eq:Q_ui} says that the fluid amount of class-$(u,i)$ jobs increases with rate $\lambda d p_{0,i} \pp(\eta_1> \tau_i\mu_i)$ due to external arrivals and decreases by $\frac{( \T_{u,i}(t) - \V_{u,i} )^+}{\tau_i}$ as these jobs stay in queue $i$ for $\tau_i$ time units.
Similar interpretations are easily obtained for $\Q_{c,i}(t)$ and $\Q_{j,i}(t)$.
Equation~\eqref{eq:FSL} is therefore interpreted as the work-conserving condition.

\begin{remark}
\label{rem00}
The set of fluid solutions associated to a speculative queueing network is identical to the set of fluid solutions associated to its corresponding multiclass queueing network where the routing and service processes are independent; see Formulas (4.17)-(4.21) in~\cite{dai1995positive}.
\end{remark}

From the above remark and from Theorem 4.1 of \cite{dai1995positive}, it follows that the fluid model equations are the limit equations satisfied by a speculative queueing network. We refer the reader to Theorem 4.1 for a formal statement and proof of this fluid limit argument, from which our stated fluid model follows as a direct consequence. 

%


We now define stability for the \emph{fluid model}.
\begin{definition}
\label{def2}
We say that the fluid model is \emph{stable} if there exists a constant $\delta>0$ such that 
for any fluid solution with $\|\Q(0)\|+\|\UL\|+\|\V\|=1$, it holds that
$\Q(t+\delta)=0$ for all $t >0$.
\end{definition}
The constant $\delta$ above may depend on all model parameters
but not on the initial state. 

%
%


The following theorem 
connects fluid model stability with $X(t)$ being positive Harris recurrent
and was first proved in \cite{dai1995positive}.
\begin{proposition}[Dai \cite{dai1995positive}; Bramson \cite{bramson2008stability}]
\label{thDAI}
Assume that \eqref{eq:FS1}-\eqref{eq:FSL} hold. If {the fluid model solutions is stable}, then $X(t)$ is positive Harris recurrent.
\end{proposition}

The proposition follows from Theorem 4.16 of Bramson \cite{bramson2008stability}. In the Appendix, we give a proof that essentially explains how our model fits  within the framework in \cite{bramson2008stability}. We note that there are some aspects of our model that are non-standard within the usual fluid stability framework: i) jobs are routed according to their size rather than according to an independent mechanism, and ii) the size of jobs within a queue can have a different probability distribution (jobs having non-identical distributions within a multi-class queue can impact stability, see the well-known counter-example of Bramson \cite{bramson1994}). However, in our case, it is possible to segment arrivals into an extended set of job classes as described in Section~\ref{sec:multiclass_rep}. This yields the required independence properties for routing probabilities for jobs within each class. This deals with part i) above. For part ii), we note that fluid stability remains sufficient for positive recurrence so long as we have a Markovian state description and head-of-the-line service within each class. Each of these apply to our model and therefore what remains is to prove the stability of the fluid model. This result is stated~below.

%

\subsection{Fluid Stability and Positive Harris Recurrence}

We now state our main result on stability. 
Let
\begin{equation}
\label{eq:rho_def}
\rho_i 
\bydef \lambda N p_{0,i} \E\left[\frac{\eta_1}{\mu_i}\wedge \tau_i\right]   +
 \lambda N \sum_{j=1}^N 
p_{0,j} p_{1,i} 
\pp(\eta_1>\mu_j\tau_j) 
\E\left[\frac{\eta_2}{\mu_i}\,\Big|\,\frac{\eta_1}{\mu_j}>\tau_j\right] 
\end{equation}
be the \emph{nominal load} of queue~$i$, for all~$i$. This accounts for the work from both speculative and non-speculative jobs.

The following result shows that the Markov process of interest is stable under the condition that $\rho_i<1$ for all $i$.
Our proof is based on Proposition~\ref{thDAI} and on a Lyapunov argument; see the Appendix.

\begin{theorem}
\label{th1}
For all $i$, assume that $\rho_i<1$.
Then, $X$ is positive Harris recurrent. 
\end{theorem}

Within the proposed queueing network, the mean service time of a job depends on the number of visits currently performed in the network and that the network topology is not feedforward. In these cases, it is well known that the usual stability condition is in general not sufficient to make the underlying Markov process positive Harris recurrent. A counterexample with Poisson arrivals, exponentially distributed service times and two FCFS queues is the reentrant line developed in \cite{bramson1994}; see also \cite{328805,gamarnik2005}. 

We note that the stability condition in Theorem~\ref{th1} differs from the ``usual'' condition $\lambda N p_{0,i} / \mu_i <1$, which is the nominal stability condition when no timeout is set (no speculation). It is thus natural to investigate which approach yields the largest stability region. This is one of the objectives of the following section.



\section{Timeout Design}
\label{sec:timeout_design}

In this section, we rely on Theorem~\ref{th1} to perform some optimization and understand the advantages of speculative load balancing.
Specifically, since timeouts and routing probabilities are under the control of the network manager, 
we investigate whether or not it is possible to design them in a manner such that the resulting stability region is larger than the stability region achieved by 
\begin{itemize}
\item[i)] standard load balancing; i.e., no replication and no speculation, though servers and dispatcher(s) can exchange dynamic control messages as in join-the-shortest-queue;
 
 \item[ii)] other replication schemes based on the principle ``replicate upon job arrivals''.
\end{itemize}

\begin{remark}
We will address the question above when servers are {homogeneous}, that is $\mu_i=1$ for all $i$.
The heterogeneous case can be handled as well but at the cost of complicating the exposition unnecessarily.
For this reason, we limit this application of Theorem~\ref{th1} to the homogeneous case only.
\end{remark}

Given $\tau\in\mathbb{R}_+$, let us define 
\begin{equation}
\label{eq:rho_def_symmetric}
\rho(\tau) \bydef \lambda \E[\eta_1\wedge\tau] + \lambda \pp(\eta_1>\tau) \E[\eta_2\mid\eta_1>\tau]
\end{equation}
and note that $\rho(\tau)=\rho_i$ if the routing probabilities are identical and all timeouts are equal to $\tau$.

The following proposition says that if servers are homogeneous, then it is optimal to choose identical timeouts and routing probabilities; see the Appendix for a proof.

\begin{proposition}
\label{prop:symmetric}
If $\mu_i=1$ for all $i$, then
\begin{equation}
\inf_{\tau\in\mathbb{R}_+} \rho(\tau) = \inf \max_{i} \rho_i
\end{equation}
where $\rho(\tau)$ is defined in \eqref{eq:rho_def_symmetric}
and
the second $\inf$ is taken over all stochastic routing vectors $p_0,p_1\in\mathbb{R}_+^N$ and timeout vectors~$\tau\in\mathbb{R}_+^N$.
\end{proposition}


In view of Proposition~\ref{prop:symmetric}, in the analysis that follows we will require the following assumption, which implies that the speculative queueing network is symmetric.

\begin{assumption}[Symmetric network]
\label{symmetric}
For all $i$, $\mu_i=1$, $p_{0,i}=p_{1,i}=\frac{1}{N}$ and $\tau_i=\tau$, for some $\tau\in\mathbb{R}_+$.
\end{assumption}

In the following, we first
investigate scenarios where timeouts help to increase the stability region with respect to standard load balancing. Then, we characterize an optimal timeout and provide numerical simulations to compare the performance of speculation with existing replication schemes.

\subsection{Speculation vs Standard Load Balancing}
\label{sec:comp_bernoulli}

Under Assumption~\ref{symmetric},
we investigate when the proposed approach yields an increased stability region with respect to 
standard load balancing.
 In the latter case, it is well known that $\lambda \E[\eta_1]<1$ is the necessary and sufficient stability condition for the Markov process that models the dynamics induced by several dispatching algorithms such as join-the-shortest-queue. 
Therefore, given Theorem~\ref{th1}, we aim at finding $\tau\in\mathbb{R}_+$ such that
\begin{equation}
\label{eq:design_tau}
\rho(\tau) < \lambda \E[\eta_1].
\end{equation}
This problem can be easily addressed at least numerically once probability distributions for the $\eta_i$'s are known. Nonetheless, in this section our aim is to find insights and general conditions ensuring that~\eqref{eq:design_tau} holds.

We have the first general condition; see the Appendix for a proof.
\begin{theorem}
\label{th3}
Let Assumption~\ref{symmetric} hold. Then,
$\rho(\tau) < \lambda \E[\eta_1]$
if and only if
\begin{equation}
\label{condition1}
\E[\eta_2\mid \eta_1>\tau]  < \E[\eta_1 - \tau\mid \eta_1> \tau].
\end{equation}
\end{theorem}

The inequality \eqref{condition1} admits a simple interpretation because the RHS term is the expected remaining service time of the job in progress after age~$\tau$ and the LHS term is the expected service time of a second execution if the job would timeout at~$\tau$.

%

We now derive a further condition introducing some structure 
on the service time random variables~$\eta_i$, $i=1,2$.
In the following, $S$, $S_1$, $S_2$ and $X$ are four auxiliary nonnegative random variables each independent of all else and such that~$S$, $S_1$ and~$S_2$ are equal in distribution.

\begin{assumption}[The ``$S\&X$'' model]
\label{as_SX}
For $i=1,2$, $\eta_i=S_i X$.
\end{assumption}
The service time structure in Assumption~\ref{as_SX} was first introduced in~\cite{bettermodelTON}:~$X$ is interpreted as a job's \emph{intrinsic size} and $S_i$ represents the \emph{server slowdown} incurred by a job's $i$-th execution.
An important special case is obtained when~$X$ is deterministic, which models systems where the service time variability is only due to server runtime phenomena. Here, the service times associated to the execution of a job are independent; a line of papers investigate this case, e.g., \cite{gardner2017redundancy,Raaijmakers2019}.
\begin{remark}
\label{rem1}
Assumption~\ref{as_SX} models the case where jobs are re-executed from scratch after re-routing and
rules out the possibility that they may be resumed.
This is a worst case in our respect because speculation (and our model) allows jobs to be resumed after re-routing, and the stability region obtained by resuming jobs is clearly increased. 
\end{remark}

We will also make use of the following assumption.
\begin{assumption}
\label{as_z}
For some $z\in\mathbb{R}_+$,
\begin{align}
\label{as1}
\E\left[S x \wedge z \right] <\, \pp\left(S x\le z\right)\, \E[S]\,x, \quad\forall x\in{\rm support}(X).
\end{align} 
\end{assumption}

The following result, proven in the Appendix, provides a sufficient condition for which speculation yields an increased stability region with respect to standard load balancing.

\begin{theorem}
\label{th_cond2}
Let Assumptions~\ref{symmetric} and~\ref{as_SX} hold. Then, $\rho(z) < \lambda \E[\eta_1]$ for all $z\in\mathbb{R}_+$ that satisfy Assumption~\ref{as_z}.
Furthermore, if $X$ is deterministic, then Assumption~\ref{as_z} is also necessary to have $\rho(z) < \lambda \E[\eta_1]$.
\end{theorem}

Theorem \ref{th_cond2} shows that Assumption \ref{as_z} is almost the requirement for speculation to increase the size of the stability region and we now briefly discuss it. 
Towards this purpose, let~$\mathcal{T}_{S,X}$ be the set of~$z\in\mathbb{R}_+$ such that \eqref{as1} holds. 
Assuming that $X=1$, it is not difficult to show that:
\begin{enumerate}
 \item If $S$ is $s_m$ with probability $p$ and $s_M>s_m$ with probability $1-p$ (a bimodal  distribution), then $\mathcal{T}_{S,X}=(s_m,(s_M-s_m)p)$; typically, $s_M\gg s_m$ and $p\approx 0.99$ \cite[Section~2.6.4]{Barroso2009}, and therefore $\mathcal{T}_{S,X}$ is not empty in practice.

 \item If $S$ follows a nondegenerate hyperexponential distribution, then $\mathcal{T}_{S,X}=(0,\infty)$;

 \item If $S \sim$ Pareto$(\alpha,s_m)$, where $\alpha>1$ and 
$s_m$ respectively denote the shape and scale parameters, 
then $\mathcal{T}_{S,X}=(\alpha s_m,\infty)$.
\end{enumerate}
On the other hand, if~$X=1$ and~$S$ follows the ``more deterministic'' Erlang distribution, then one can show that $\mathcal{T}_{S,X}$ is empty and in this case Theorem~\ref{th_cond2} implies that speculation loses capacity.

Now, let 
\begin{equation}
\label{eq:LR}
L(\tau) \bydef  \frac{\rho(\tau)}{\lambda \E[\eta_1]}.
\end{equation}
If $L(\tau)<1$, then $L(\tau)$ can be interpreted as the \emph{load reduction} with respect to standard load balancing when adopting timeout $\tau$.
As in Section 2.6.4 of \cite{Barroso2009}, let us assume that
\begin{equation}
\label{S_bimodal}
S=\left\{
\begin{array}{ll}
10 & \mbox{ w.p. } 0.99\\
10^3 & \mbox{ w.p. } 0.01,
\end{array}
\right.
\end{equation}
i.e., $S$ follows a bimodal distribution.
By increasing the timeout $\tau$, Figure~\ref{fig:plot_ell} plots $L(\tau)$
under a number of combinations on the distributions of~$S$ and~$X$.
For the Pareto distribution, we have chosen $s_m=1$ and $\alpha\in\{1.1,1.5\}$; it is known that $\alpha\in[1,1.5]$ is the most common range~\cite{Ananthanarayanan13,USENIX12,HarcholBalter2009}. 
For the hyperexponential distribution, we have chosen two phases with balanced means, as in \cite{HarcholBalter2009}. Specifically, the probability and rate vectors are $[0.99, 0.01]$ and $[1,1/99]$, respectively.
\begin{figure*}
\hspace*{-0.8cm}
\centering
\includegraphics[width=18cm]{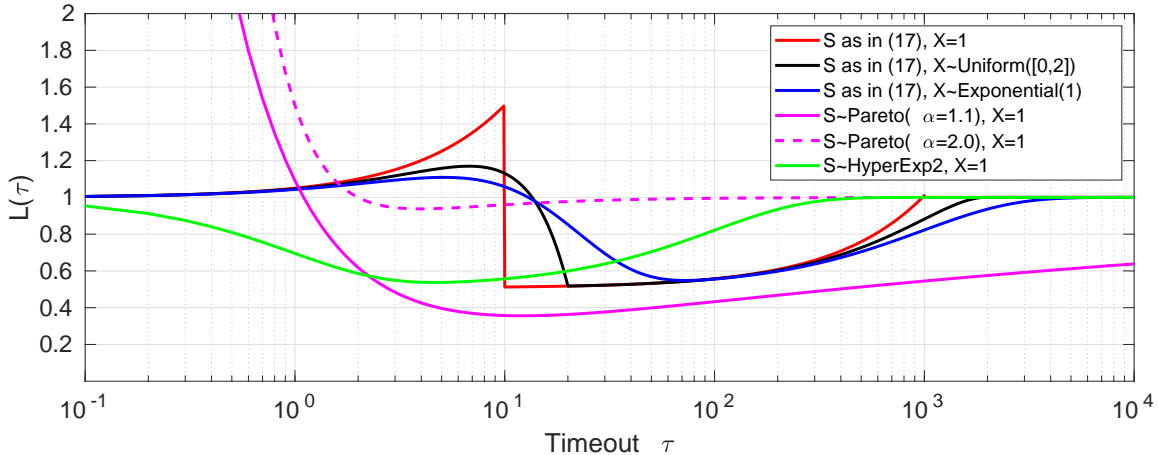}
\caption{Speculative vs standard load balancing via $L(\tau)$, \eqref{eq:LR}, under a number of $S\&X$ models.}
\label{fig:plot_ell}
\end{figure*}
%
%
If timeouts are too small, Figure~\ref{fig:plot_ell} shows that speculation may loose capacity ($L(\tau)>1$), and this is to be expected because jobs have almost no chance of completing before the timeout.
If timeouts are large enough, however, they increase the size of the stability region ($L(\tau)<1$).
We notice that the timeouts that minimize the load can double the size of the stability region achieved by standard load balancing because within such timeouts $L \approx 0.5$.
When $S$ has the bimodal form given in \eqref{S_bimodal},
this benefit is obtained whenever 
$\tau\in[10,10^3)$ (if $X=1$),
$\tau\in[13,2\times 10^3)$ (if $X\sim$ Uniform([0,2])), and
$\tau\in[13.5,\infty)$ (if $X\sim$ Exponential(1)).
When $S$ is Pareto, we observe that $L(\tau)$ decreases when $\alpha$ decreases for any given $\tau$, i.e., timeouts are particularly helpful when the service time variability increases.
When $S$ is hyperexponential, we also observe that all timeouts reduces the load. This is in agreement with the statement in point 2) above.

Since the set of $\tau$'s ensuring that $L(\tau)<1$ is somewhat large,
we conclude that 
replication by speculation is robust with respect to small errors in service time distributions and thus that 
the system manager has some flexibility when designing a proper timeout rule.

\subsection{An Optimal Timeout}
Our stability result, Theorem \ref{th1}, demonstrates that the load, $\rho(\tau)$, is the key first order performance indicator for a speculative network. 
Theorems \ref{th3} and \ref{th_cond2} show that timeouts can reduce load and thus increase capacity.  
Now, we investigate the minimizers of the load function $\tau\mapsto \rho(\tau)$, i.e., we characterize  the optimal timeout.
We address this problem under the following assumption.
\begin{assumption}\label{OTas}
We assume that $\eta_1$ has a probability density function $f_1(t)$ with a decreasing hazard function, i.e.,
\begin{align}\label{Hazard}
  t\mapsto \frac{f_1(t)}{\int_t^\infty f_1(s) ds }
\end{align}
is decreasing. Further, we assume that  
$\bareta_2(t):=  \mathbb E [ \eta_2 | \eta_1 > t] $ is such that
\begin{align}\label{mu2cond}
  t \mapsto \frac{1+ \bareta_2'(t)}{\bareta_2(t)}
\end{align}
is non-decreasing. 
\end{assumption}
The function~\eqref{Hazard} is the hazard function of a probability distribution and it is well-known that a decreasing hazard function is a characteristic of a heavy-tailed  probability distribution. 
A trivial condition for \eqref{mu2cond} to be non-decreasing is that $\eta_2$ is independent of $\eta_1$. 
This corresponds to the case where job sizes are fixed but variability in service occurs due to independent random events at each server. 
In this case, we note that \eqref{Hazard} states that the complementary CDF of $\eta_1$ is log-convex. This complements with the findings of \cite{joshi2017efficient} where it is found that independent log-convex distributions service time distributions benefit from cancel-on-complete replication. 

Let us define the optimal timeout as follows.

\begin{definition}[Optimal Timeout]\label{OptDef}
The optimal time $\tau$ is the smallest time such that
\begin{align}\label{stop}
  \frac{f_1(\tau)}{\int_\tau^\infty f_1(s)ds } \leq \frac{1+ \bareta_2'(\tau)}{\bareta_2(\tau)} \,.
\end{align}
\end{definition}
If we assume that $\eta_2$ is independent of $\eta_1$, then we notice that the optimal timeout rule reduces to the condition
\begin{align} \label{indyTime}
    \frac{f_1(\tau)}{\int_\tau^\infty f_1(s)ds } = \frac{1}{\mathbb E[\eta_2]}.
\end{align}

Definition~\ref{OptDef} provides a practical rule for speculation because it allows one to apply reinforcement learning or statistical estimation techniques to learn the optimal timeout when the service time distributions are not known in advance.


The following is our main result on optimal timeout design; see the Appendix for a proof.
\begin{theorem}
\label{th_optimal_timeout}
Let Assumptions~\ref{symmetric} and~\ref{OTas} hold.
Finite optimal timeouts (see Definition~\ref{OptDef}) exist and minimize the load $\rho(\tau)$ induced on a speculative queueing network. Moreover, if, in addition, the service requirements $\eta_1$ and $\eta_2$ are independent, then any value of $\tau$ satisfying~\eqref{indyTime} minimizes the load.
\end{theorem} 

\subsection{Speculation vs Replication}
\label{rep_vs_spec}

We now compare the performance achieved by Speculative Load Balancing (SLB) with the performance achieved by strategies that replicate jobs at the time of their arrival in the network.
We perform such comparison within the $S\&X$ model in Assumption~\ref{as_SX}, that is a worst case scenario for speculation in view of Remark~\ref{rem1}.

First, let us consider a setting where each incoming job is replicated to~$d$ queues selected uniformly at random and independently of anything else. Here, redundant replicas are canceled as soon as one either completes (Cancel-on-Complete) or starts (Cancel-on-Start) service.
We refer to the former resp. latter scheme as CoC-$d$ resp. CoS-$d$.


It is known that CoS-$d$ is equivalent to the Least-Left-Workload-$d$ (LLW-$d$) dispatching algorithm. We recall that LLW-$d$ sends each incoming job to a queue having the shortest remaining workload among~$d$ selected uniformly at random, with ties broken randomly.
Since stability for CoS-$d$ is obtained if and only if $\lambda \E[\eta_1]<1$, i.e., as in standard load balancing, we have the following remark.
\begin{remark}
SLB yields an increased stability region w.r.t. CoS-$d$ whenever the conditions in Theorems~\ref{th3} and~\ref{th_cond2} are satisfied. 
\end{remark}

On the other hand, it is difficult to perform a comparison with CoC-$d$ because a satisfactory characterization of the resulting stability region is currently unclear
\cite{Borst2020,MENDELSON2021113,Raaijmakers2019}.
However, depending on the load regime, we can argue intuitively as follows:
\begin{itemize}
\item Assume the network is \emph{lightly loaded}, or over-provisioned. Then, the time spent in the system by a job with CoC-$d$ is $\min(\eta_1,\ldots,\eta_d)$, where $\eta_c$ represents the service time of copy~$c=1,\ldots,d$. This should be compared to $\eta_1\wedge \tau + \eta_2\,\mathbb{I}_{\{\eta_1>\tau\}}$, the time spent by a job in a speculative queueing network within the same load condition. For any $d\ge 2$, it is not difficult to show that\footnote{We recall that $\eta_1$ and $\eta_2$ are assumed equal in distribution in this section.} 
$\min(\eta_1,\ldots,\eta_d) \le \eta_1\wedge \tau + \eta_2\,\mathbb{I}_{\{\eta_1>\tau\}}$, which means that, in a lightly loaded regime, it is better to replicate upon job arrivals rather than speculate upon straggler detection. This claim should also be intuitive.

\item Assume the network is \emph{moderately loaded}.
Then, within CoC-$d$ it may happen that the first copy that completes is not the minimum and that the fast copy gets stuck in its queue.
For instance, for a given job $n$, let us assume that $\eta_2(n)\ge \eta_1(n)$ and that the copy that completes first when applying CoC-$d$ is the one of size $\eta_2(n)$. On the event $\eta_1(n)\le \tau$, speculation on that job induces a lower load and no price is paid for the extra copies.

\item Assume the network is \emph{heavily loaded}.
When $\lambda$ increases, the scenario depicted in the moderate load regime amplifies, potentially leading CoC-$d$ to be unstable.
Here, it is not intuitive which approach is better than the other.


\end{itemize}

To support the intuition above, we present the results obtained by running a large set of numerical simulations.
In our tests, we assume the service time structure given in Assumption~\ref{as_SX}.
Figure~\ref{fig:plot_SLB_vs_RED} plots the average response time (time spent in the system) obtained within SLB, CoC-$d$ and CoS-$d$, for $d=2,4$ by increasing $\lambda\E[\eta_1]$ while keeping $\E[\eta_1]$ constant and under a number of $S\&X$ models where $\E[X]=1$ and the random variable $S$ is given by \eqref{S_bimodal} (bimodal distribution) or follows a Pareto distribution with shape parameter $\alpha\in\{1.1,1.5,2\}$ over the support~$[1,10^3]$.
For SLB, we have chosen the timeouts $\tau^*$ specified in the respective subfigures.
Each point ($\ast$) in each curve refers to an average of 50 simulations and each simulation executes $10^7$ jobs.
We also assume $N=50$ FCFS queues.

\begin{figure*}[t]
\hspace*{-2.0cm}
 \centering
\includegraphics[width=20cm]{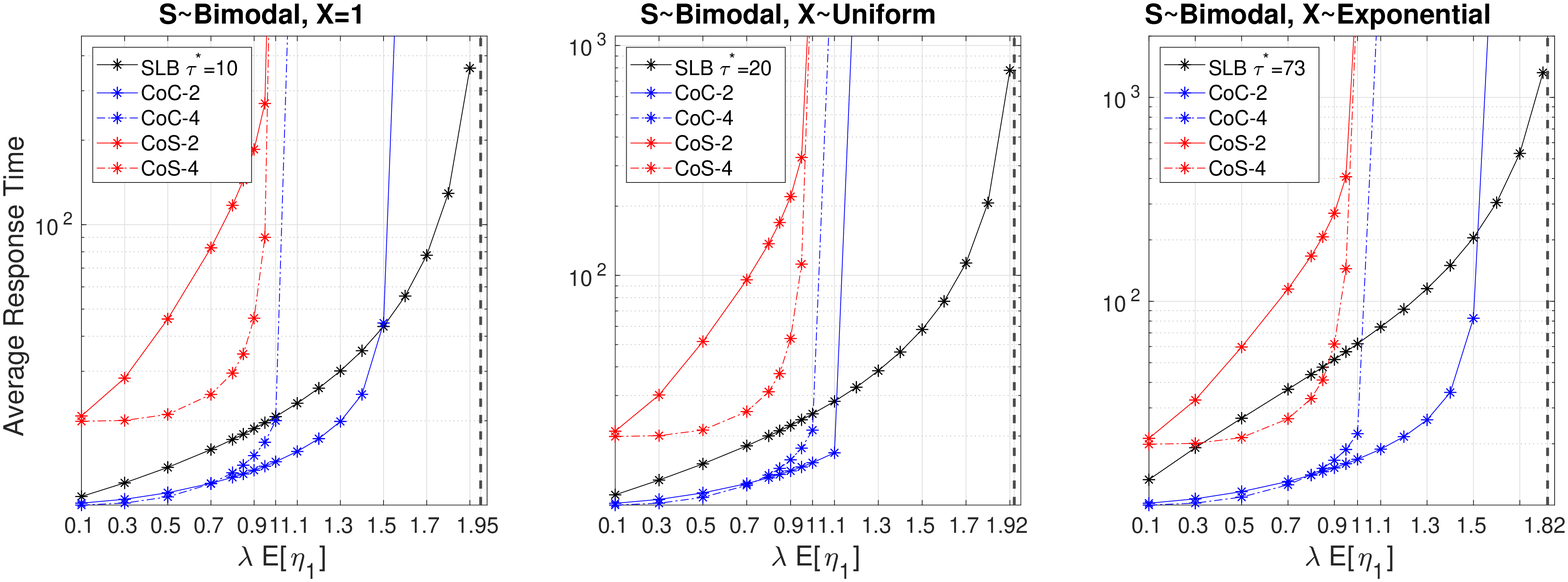}
\\
\hspace*{-2.0cm}
 \centering
 \includegraphics[width=20cm]{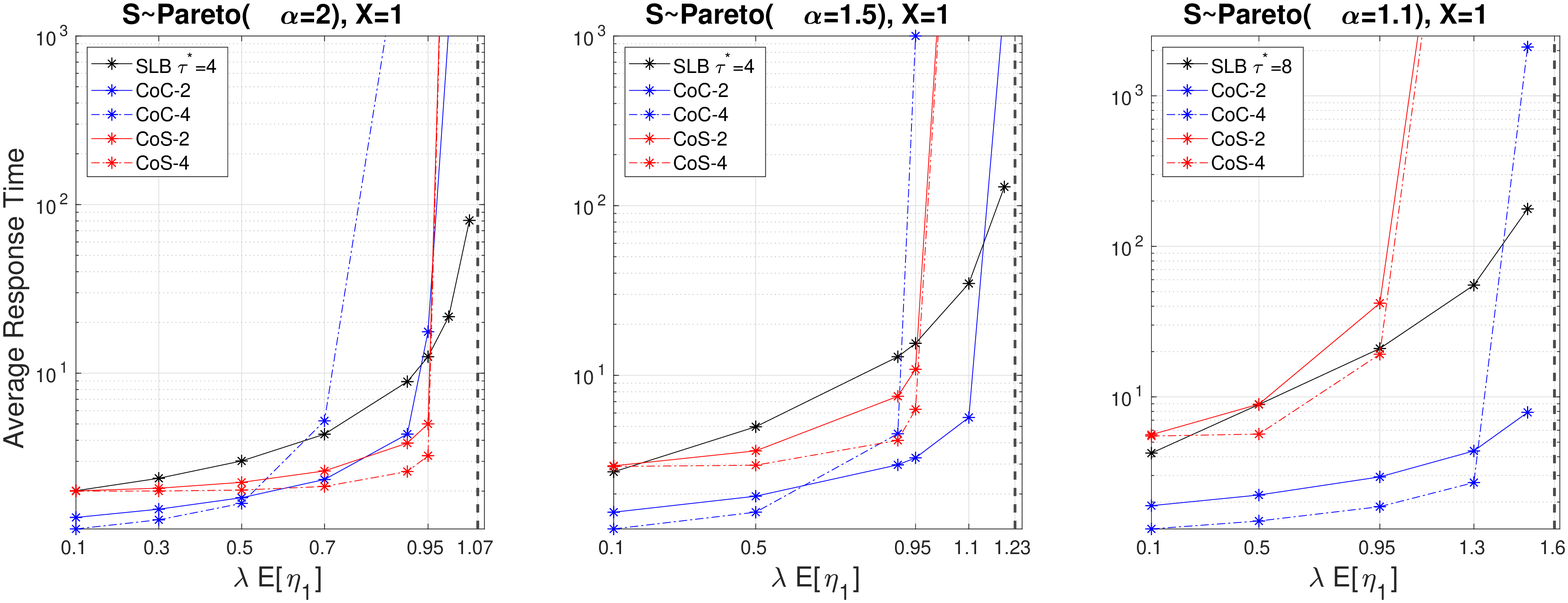}
\caption{Average response time obtained within Speculative Load Balancing (SLB), Cancel-on-Complete-$d$ (CoC-$d$) and Cancel-on-Start-$d$ (CoS-$d$) under $S\&X$ models.
The vertical dashed black lines represent the limits of the stability region of SLB.
}
\label{fig:plot_SLB_vs_RED}
\end{figure*}

The plots in Figure~\ref{fig:plot_SLB_vs_RED} confirm the intuition above.
CoC-$d$ provides the best results in light load conditions (as expected) and it also increases the stability region achieved by CoS-$d$, that is the stability region of standard join-the-shortest-queue-like algorithms ($\lambda\E[\eta_1]<1$).
However, in most cases, SLB goes further and is able to accommodate more traffic than CoC-$d$.
For the most extreme heavy tailed distribution, Pareto with $\alpha=1.1$, CoC-$2$ eventually provides the best results, though within CoC-$d$ is not clear how to well choose~$d$ \emph{a priori}. 

We now compare SLB with Redundant-to-Idle-Queue-$d$ (RIQ-$d$),
a replication scheme that works as CoC-$d$ but replicas are only made to those servers which are idle, and if no idle server is found then the job is sent to a random one of the $d$ selected; see~\cite{bettermodelTON}.
RIQ-$d$ was introduced to avoid the potential loss of capacity of CoC-$d$, though at the same time we observe that it can ban its potential gain.
We also notice that comparing SLB and RIQ-$d$ is not completely fair because the latter scheme is \emph{dynamic} in the sense that it needs dispatcher(s) and queues to continuously exchange information about servers' status. 
No feedback mechanisms to dispatcher(s) are assumed in SLB, though they could clearly be integrated to further reduce delays; we do not discuss this in this paper.
Within the same setting described above,
Figure~\ref{fig:plot_SLB_vs_RED} plots the average response time obtained within SLB and RIQ-$d$, for $d=2,4$ by increasing $\lambda\E[\eta_1]$ while keeping $\E[\eta_1]$ constant. 
\begin{figure*}[t]
\hspace*{-2.0cm}
 \centering
\includegraphics[width=20cm]{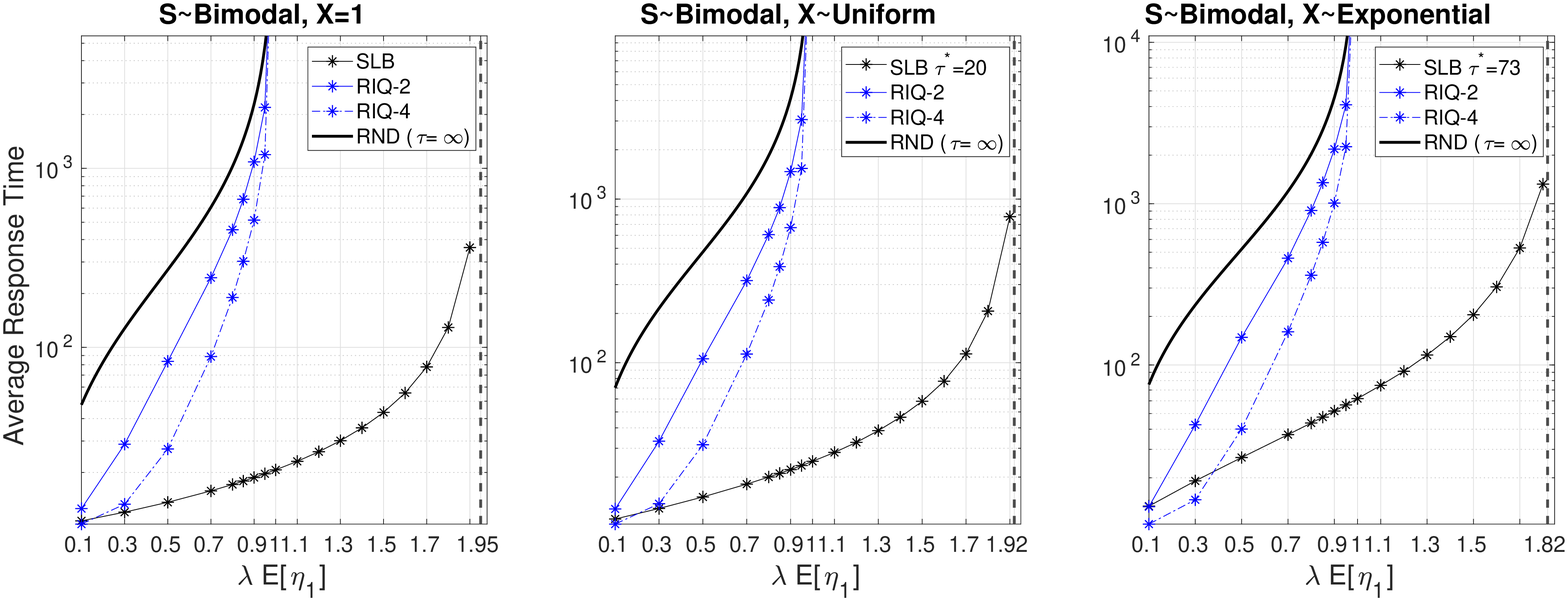}
\caption{Average response time obtained within Speculative Load Balancing (SLB) and Redundant-to-Idle-Queue-$d$ (RIQ-$d$) under $S\&X$ models.
The vertical dashed black lines represent the limits of the stability region of SLB.
}
\label{fig:plot_SLB_vs_RIQ}
\end{figure*}
RIQ-$4$ provides slightly better results in light load conditions, which is to be expected because as $\lambda\downarrow 0$ it behaves as CoC-$4$. As the load increases, however,
it is less and less likely to find idle queues and the dynamics of RIQ-$d$ get closer and closer to the dynamics of Random (RND), which sends each job to a single server selected independently at random.
For the latter, the stability condition $\lambda \E[\eta_1]<1$ applies, while SLB preserves stability much~further.

Finally, we conclude this section with the following remark about the communication overhead induced by SLB, CoC-$d$, CoS-$d$ and RIQ-$d$. SLB requires $1+\pp(\eta_1> \tau)<2$ control messages per job in average. On the other hand, CoC-$d$ and CoS-$d$ require, for each job, $d$ dispatch messages plus $d-1$ cancellation messages.

\section{Response Time for Large Systems}\label{sec:Large_Response}

In this section, we consider a symmetric (see Assumption~\ref{symmetric}) speculative queueing network composed of FCFS queues. 
Within these assumptions, let $R_{N,n}(\tau)$ be the overall time spent by the $n$-arriving job in the system and let
\begin{align}
R_N(\tau) \bydef \limsup_{m\to\infty} \frac{1}{m}\sum_{n=1}^m \E[R_{N,n}]
\end{align}
be the average \emph{response time} experienced by jobs.
In this section, our goal is to investigate~$R_N(\tau)$.

First, we observe that $R_N(\infty)$ corresponds to the mean response time of an M/GI/1 queue.
%
When~$\tau<\infty$, the feedback speculation mechanism significantly complicates the analysis and our aim is to develop an approximation.
We focus on the large system limiting regime where $N\to\infty$ and $\rho(\tau)$ is constant.
Defining
\begin{align}
W&\bydef \frac{\lambda}{2} (1+\pp(\eta_1\ge \tau)) \frac{M}{1-\rho(\tau)}\\
%
M&\bydef  \frac{\E[(\eta_1\wedge \tau)^2] + \E[\hat{\eta}_2^2]\, \pp(\eta_1>\tau)}{1+\pp(\eta_1>\tau)}
\end{align}
where $\hat\eta_2$ is an auxiliary random variable equal in distribution to $\eta_2\mid\eta_1>\tau$, we claim that the following conjecture holds true.

\begin{conjecture}
\label{conjecture}
Provided that $\rho(\tau)<1$,
\begin{align}
\label{lim_responsetime}
\lim_{N\to\infty} R_N(\tau)
%
%
& = (1 + \pp(\eta_1> \tau)) W  + \frac{\rho(\tau)}{\lambda}.
\end{align}
\end{conjecture}

The underlying justification behind Conjecture~\ref{conjecture} lies in postulating that queues become ``asymptotically independent'' in the limit $N\to\infty$.  
In this case, the arrival process at each queue~$i$ is the superposition of an exogenous rate-$\lambda$ Poisson process and~$N$ independent feedback processes each with rate $\lambda\, \pp(\eta_1>\tau)/N$.
Given that the intensity of each of the feeback processes approaches zero as $N\to\infty$, we may assume that their superposition is ``approximately'' a Poisson process with rate $\lambda\, \pp(\eta_1>\tau)$, provided that $N$ is ``large''.
This is justified by the Palm--Khinchine theorem, which ensures that the superposition of several independent  sparse point processes converges weakly to a Poisson process on~$\mathbb{R}_+$; see \cite[Chapter~5.8]{sobel}.
%
Therefore, we conjecture that the arrival process at each queue~$i$ is the independent superposition of 
1) a rate-$\lambda$ Poisson process carrying jobs of sizes equal in distribution to $\eta_1$ and
2) a rate-$\lambda\, \pp(\eta_1>\tau)$ Poisson process carrying jobs of sizes equal in distribution to $\eta_2|\eta_1\ge \tau$.
This queue can be interpreted as an M/G/1 queue where the arrival rate is $\lambda + \lambda\, \pp(\eta_1>\tau)$ and the service times are equal
in distribution to $H$, where 
\begin{equation}
\label{H}
H=\left\{
\begin{array}{ll}
\overline{\eta}_1\wedge\tau & \mbox{ w.p. } \frac{1}{1 + \pp(\eta_1>\tau)}\\
\eta_2\mid\eta_1\ge \tau & \mbox{ otherwise}
\end{array}
\right.
\end{equation}
with $\overline{\eta}_1$ equal in distribution to $\eta_1$ but independent of all else.
Note that $\E[H^2]=M$. Furthermore, the traffic intensity at this M/G/1 queue is
$( \lambda +\lambda \pp(\eta_1>\tau)) \E[H]=\rho(\tau)$
and the mean workload observed at arrival times is~$W$, which follows by applying the Pollaczek-Khinchine formula.
Conditioning on $\eta_1\le\tau$, the mean response time of a job~is   
\begin{align}
\label{W_cond1}
W + \E[\eta_1\mid \eta_1\le \tau],
\end{align}
and conditioning on $\eta_1>\tau$, the mean response time of a job is   
\begin{align}
\label{W_cond2}
W + \tau + W + \E[\eta_1\mid \eta_1\le \tau],
\end{align}
that is the time spent during the first visit, $W + \tau$, plus 
the time spent during the second visit, $W + \E[\eta_1\mid \eta_1\le \tau]$.
Putting~\eqref{W_cond1} and~\eqref{W_cond2} together, we get
\begin{equation*}
\pp(\eta_1\le \tau)  (W + \E[\eta_1\mid \eta_1\le \tau]) 
+  \pp(\eta_1> \tau) ( 2W + \tau + \E[\eta_2|\eta_1\ge \tau] ),
\end{equation*}
which after some algebra boils down to~\eqref{lim_responsetime}.

Figure~\ref{fig:conjecture} plots the average time spent by a job in a speculative queueing network
obtained by simulation ($N=50$) and by the conjectured Formula~\eqref{lim_responsetime}.
\begin{figure}
\centering
\includegraphics[width=13cm]{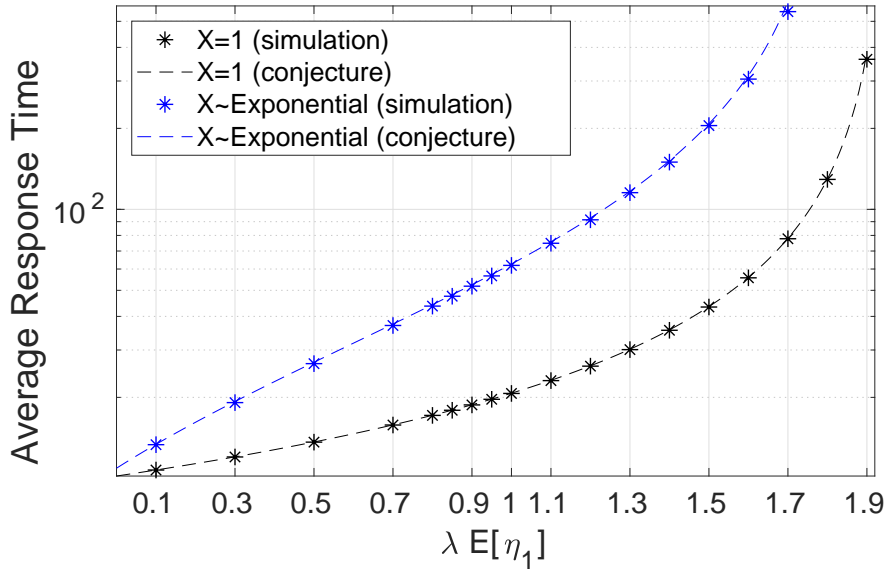}
\caption{Average response times obtained by simulation and via Conjecture~\ref{conjecture}; $S$ has the form \eqref{S_bimodal} and $N=50$.}
\label{fig:conjecture}
\end{figure}
Each point ($\ast$) of each plot refers to an average of 50 simulations and each simulation was based on $10^7$ jobs. 
As in the plots of Section~\ref{rep_vs_spec}, we assume $\eta_1=S_1 X$ and $\eta_2=S_2X$ where $S$ is as in~\eqref{S_bimodal}. For $X$, we distinguish the cases $X=1$ and $X\sim$ Exponential(1). It is remarkable how the approximation is accurate under any load condition.

\section{Conclusions}
\label{sec:conclusions}

In this paper, we provide a first performance evaluation for job speculation in a queueing network. 
We characterize the stability region of a speculative queueing network and find the resulting stability region to be distinct from a queueing network implementing a standard load balancing scheme. We provide conditions on job size distributions for speculative load balancing to increase the size of the stability region. Specifically, in the presence of heavy tailed server slowdown, we find that speculation is a good mechanism to increase stability and thus throughput. We provide the first characterization of the optimal timeout for a speculative task. This follows simple implementable formulas~\eqref{stop} and, under independence,~\eqref{indyTime}. 
Under moderate to heavy loadings, simulations indicate that speculation can significantly improve performance when compared with redundancy schemes and standard load-balancing systems. 
Finally, using the Pollaczek-Khinchine formula, we postulate on the impact of speculation on response times by providing a numerically accurate approximation.

The proposed framework opens up important questions:
\begin{itemize}
 \item 
It is possible to design schemes that combine the benefits of replication for a light loaded regime, while maintaining the desirable stability properties of speculation for moderate to high loads.
 \item 
The design of the optimal timeout requires the solution of a Markov decision process. Since task sizes must be statistically estimated, real systems may need to apply reinforcement learning to design optimal task based replication decisions.

 \item 
The analysis in Section \ref{sec:Large_Response} forms a conjecture on the mean field limit of a speculative queueing network. The resolution of this conjecture requires further analysis.

 \item 
We have applied random server assignment to arriving and speculative tasks. However, it may be preferable to implement join-idle-queue or join-the-shortest-queue-$d$ routing to these tasks. While we believe that the stability region obtained within such speculative dynamic schemes remains unchanged, delays are expected to be further reduced.

\end{itemize}

A natural generalization of our model considers multiple speculation levels. Specifically, timeouts can be used in second job executions and a job that timeouts in the second visited queue gets routed to a third queue, and so forth a number of times. We believe that our results generalize to this setting naturally, though a formal proof for Theorem~\ref{th1} would require more work.

Although job speculation is widely used in practice, its performance has been understudied relative to recent works on redundancy. The proposed framework sheds a new light on speculative load balancing and addresses a number of important questions on the design of job replication in large scale computer systems.



  \section*{Acknowledgment}

This work is supported by the French National Research Agency in the framework of the ``Investissements d'avenir'' program (ANR-15-IDEX-02) and the LabEx PERSYVAL (ANR-11-LABX-0025-01).

\bibliographystyle{abbrv}

\section*{Proof of Proposition \ref{thDAI}}

The proof below identifies how our model fits within the framework of Bramson \cite{bramson2008stability}, under the multi-class structure given in Section \ref{sec:multiclass_rep}.

\begin{proof}[Proof of Proposition \ref{thDAI}] 
The proof follows from arguments in Bramson \cite{bramson2008stability}. 
Here we confirm that conditions stated in \cite{bramson2008stability} apply to our model and we refer the reader to appropriate sections and results in that work. 

We note that, typically in the fluid analysis of queueing networks, it is assumed that the arrival, service and routing processes are independent. This is not true in our case under our first state description in Section \ref{sec:model}. However, under the extended class structure in Section \ref{sec:multiclass_rep}, the state description does have independent job sizes (within each class). 
Further it should be noted that notions of positive Harris recurrence and petite sets remain without these independence assumptions (see the remark following \cite[Proposition 2.1]{dai1995positive}). 
First we note that the process is a (Harris) Markov process. In particular, the state descriptor given in \eqref{QUV} and given by Bramson \cite[Section 4.1 \& 4.3]{bramson2008stability} remains Markov in our setting, since when we condition on this state descriptor all future events are  function of independent random variables.
Second, we assume head-of-the-line service.
These conditions ensure that closed sets are petite; see \cite[Proposition~4.8]{bramson2008stability}. 

Finally, we observe that the conditions of \cite[Theorem 4.16]{bramson2008stability} have now been met. Thus as concluded in the theorem, queueing network stability (positive Harris recurrence) holds whenever the associated fluid model is stable. 
%
\end{proof}

\section*{Proof of Theorem~\ref{th1}}

In view of Lemma~5.3 in \cite{dai1995positive}, we can restrict to fluid solutions where $\V=0$.
Therefore, we can assume that for all $i$
\begin{align}
\label{eq:FS1_red}
&\Q_{c,i}(t) = \Q_{c,i}(0)  + \lambda N p_{0,i} \pp(\eta_1\le \tau_i\mu_i) t  - \frac{ \T_{c,i}(t) }{\E[\frac{\eta_1}{\mu_i}\mid \frac{\eta_1}{\mu_i}\le \tau_i]}  \ge 0\\
&\Q_{u,i}(t) = \Q_{u,i}(0)  + \lambda N p_{0,i} \pp(\eta_1> \tau_i\mu_i) t  - \frac{\T_{u,i}(t)}{\tau_i}      \ge 0 \\
\label{eq:FS2_red}
&\Q_{j,i}(t) = \Q_{j,i}(0)  - \frac{ \T_{j,i}(t) }{\E[\frac{\eta_2}{\mu_i}\mid\frac{\eta_1}{\mu_j}>\tau_j]} + p_{1,i}  \,\frac{\T_{u,j}(t)}{\tau_j}    \ge 0, \forall (j,i)\in\mathcal{K}_2.
\end{align}
For all $i$,
let $G_i(t)$ be such that
\begin{equation}
\label{eq:Lyapunov_i_new}
%
G_i(t)\frac{p_{1,i}}{\mu_i} =
\E\left[\frac{\eta_1}{\mu_i}\mid \frac{\eta_1}{\mu_i}\le \tau_i\right]\, \Q_{c,i}(t) 
+  \tau_i  \, \Q_{u,i}(t)\\
+ \sum_{j} 
\E\left[\frac{\eta_2}{\mu_i}\mid\frac{\eta_1}{\mu_j}>\tau_j\right]  \left(p_{1,i}\Q_{u,j}(t)  +  \Q_{j,i}(t)  \right) 
\end{equation}
and note that $G_i(t)$ is absolutely continuous.
The function $G_i(t)$ may be interpreted as a Lyapunov function on queue~$i$ and we note that it does not correspond to the workload in queue~$i$.
Substituting~\eqref{eq:FS1_red}-\eqref{eq:FS2_red} in \eqref{eq:Lyapunov_i_new}, after some algebra we obtain
\begin{align*}
G_i(t)\frac{p_{1,i}}{\mu_i}
&= \,  G_i(0)\frac{p_{1,i}}{\mu_i}+
\lambda N p_{0,i} \E\left[\frac{\eta_1}{\mu_i}\wedge \tau_i\right] t  -  \T_{c,i}(t)-  \T_{u,i}(t)  \\
&\quad +
\sum_j 
\E\left[\frac{\eta_2}{\mu_i}\mid\frac{\eta_1}{\mu_j}>\tau_j\right] 
p_{1,i} 
 \lambda N p_{0,j}\pp(\eta_1> \tau_j\mu_j) t  -  \T_{j,i}(t)  \\
 &=\,  \frac{G_i(0)}{\mu_i} -  \B_{i}(t) +
\lambda N p_{0,i} \E\left[\frac{\eta_1}{\mu_i}\wedge \tau_i\right] \,t
+
\sum_j 
\E\left[\frac{\eta_2}{\mu_i}\mid\frac{\eta_1}{\mu_j}>\tau_j\right] 
p_{1,i} 
 \lambda N p_{0,j} \pp(\eta_1> \tau_j\mu_j)\,t \\
&=\,  \frac{G_i(0)}{\mu_i} -  \B_{i}(t) + \rho_i \,t 
\end{align*}
where $\B_i(t)\bydef \T_{c,i}(t)+\T_{u,i}(t) + \sum_{j\neq i}  \T_{j,i}(t)$ is interpreted as the cumulative time queue~$i$ has been busy (non-idle) in $[0,t]$. 
Now, assuming that~$t$ is a point of differentiability for~$G_i(t)$
\begin{align}
\dot{G}_i(t) \frac{p_{1,i}}{\mu_i}
& =-  \mathbb{I}_{\{\mathcal{Q}_i(t) > 0\}} + \rho_i
 \end{align}
where $\mathcal{Q}_i(t)\bydef \Q_{c,i}(t) + \Q_{u,i}(t) + \sum_{j\neq i} \Q_{j,i}(t)$ is the total number of jobs in queue~$i$ at time~$t$.
Thus, whenever $\mathcal{Q}_i(t)>0$, 
\begin{equation}
\label{lemma32a}
\dot{G}_i(t) \frac{p_{1,i}}{\mu_i} 
=  \rho_i-1  < 0.
\end{equation}
On the other hand, if $\mathcal{Q}_i(t)=0$, $\mathcal{Q}_s(t)>0$ and $s\neq i$, then
\begin{align*}
& G_i(t) - G_s(t) 
 =  \left(\sum_{j\neq i} \E\left[ \eta_2 \mid\frac{\eta_1}{\mu_j}>\tau_j\right] \Q_{u,j}(t)\right)  -   G_s(t)  \\
& \le \left(\sum_{j\neq i} \E\left[ \eta_2\mid\frac{\eta_1}{\mu_j}>\tau_j\right] \Q_{u,j}(t)\right)  
- \E\left[ \frac{\eta_1}{p_{1,s}} \mid \frac{\eta_1}{\mu_s}\le \tau_s\right]\, \Q_{c,s}(t) 
- \frac{\tau_s \mu_s}{p_{1,s}} \, \Q_{u,s}(t)
\\
&\quad - \sum_{j\neq i} \E\left[ \eta_2 \mid\frac{\eta_1}{\mu_j}>\tau_j\right] \left(  \Q_{u,j}(t)    + \frac{\Q_{j,s}(t)}{p_{1,s}}\right)\\
& \le 
- \E\left[ \frac{\eta_1}{p_{1,s}} \mid \frac{\eta_1}{\mu_s}\le \tau_s\right]\, \Q_{c,s}(t) 
- \frac{\tau_s \mu_s}{p_{1,s}} \, \Q_{u,s}(t)
- \sum_{j\neq i} \E\left[ \eta_2 \mid\frac{\eta_1}{\mu_j}>\tau_j\right]  \frac{\Q_{j,s}(t)}{p_{1,s}} < 0
\end{align*}
where the last inequality follows because
$\mathcal{Q}_s(t)>0$.
Since $s$ is generic, we obtain
\begin{equation}
\label{lemma32b}
G_i(t) <  \min_{s\neq i} G_s(t)
\end{equation}
provided that $\mathcal{Q}_i(t)=0$.

Let $G(t)\bydef \max_{i=1,\ldots,N} G_i(t)$.
With \eqref{lemma32a} holding true when $\mathcal{Q}_i(t)>0$ and \eqref{lemma32b} holding true when $\mathcal{Q}_i(t)=0$ and $Q_s(t)>0$, we can apply Lemma~3.2 of \cite{DaiWeiss}, which implies that $G(t)$ is also an absolutely continuous nonnegative function and such that $\dot{G}(t) < -\epsilon$ for some $\epsilon>0$, provided that $t$ is a point of differentiability of $G,G_1,\ldots,G_d$ such that $G(t)>0$. Applying Lemma~2.2 of \cite{DaiWeiss}, we obtain that $G(t)=0$ for all $t\ge \delta$, for some $\delta>0$.
Finally, we notice that $G(t)=0$ implies that the system is empty.

\section*{Proof of Proposition \ref{prop:symmetric}}

We have
\begin{align*}
\frac{1}{\lambda} \max_{i=1,\ldots,N} \rho_i
& = \max_{i=1,\ldots,N}  N p_{0,i} \E\left[ \eta_1  \wedge \tau_i\right]  +
N  p_{1,i} \sum_{j} 
p_{0,j} 
\pp(\eta_1> \tau_j) 
\E\left[ \eta_2 \mid \eta_1 >\tau_j\right] \\
& \ge \sum_{i}  p_{0,i} 
\left(
\E\left[ \eta_1  \wedge \tau_i\right]   + \pp(\eta_1> \tau_i)  \E\left[ \eta_2 \mid \eta_1 >\tau_i\right]
\right)\\
& \ge \min_{i} 
\E\left[ \eta_1  \wedge \tau_i\right]   + \pp(\eta_1> \tau_i)  \E\left[ \eta_2 \mid \eta_1 >\tau_i\right]
\\
& \ge \inf_{t\in\mathbb{R}_+} 
\left(
\E\left[ \eta_1  \wedge t\right]   + \pp(\eta_1> t)  \E\left[ \eta_2 \mid \eta_1 >t\right]
\right) = \inf_{t\in\mathbb{R}_+} 
\frac{\rho(t)}{\lambda}
\end{align*}
as desired.

\section*{Proof of Theorem~\ref{th3}}
Given $\tau\in\mathbb{R}_+$, we notice that 
\begin{align*}
\frac{\rho(\tau)}{\lambda} -\E[\eta_1]
& = \E[\eta_1\wedge \tau] + \pp(\eta_1>\tau) \E[\eta_2\mid \eta_1>\tau] -\E[\eta_1]\\ 
& = \tau\,\pp(\eta_1>\tau)  + \pp(\eta_1>\tau) \E[\eta_2\mid \eta_1>\tau]  -\E[\eta_1\mid \eta_1> \tau]\pp(\eta_1> \tau) < 0 
\end{align*}
if and only if \eqref{condition1} holds.

\section*{Proof of Theorem~\ref{th_cond2}}
Using Assumption~\ref{as_SX} and that $X$, $S_1$ and $S_2$ are independent, we obtain
\begin{align*}
& \frac{\rho(\tau)}{\lambda}
 = \E[(S_1 X)\wedge \tau] + \pp(S_1 X > \tau)\, \E[S_2 X\mid S_1 X>\tau] \\
& = \sum_x \E[(S_1 X) \wedge \tau \mid X=x] \, \pp(X=x)  + \E[S_2 X I_{\{S_1 X>\tau\}}]\\
& = \sum_x \E[(S_1 x) \wedge \tau] \, \pp(X=x)   + \sum_{x} \E[S_2 x I_{\{S_1 x>\tau\}} | X=x] \pp(X=x)\\
& = \sum_x  \E\left[S_1 x\wedge  \tau \right] \, \pp(X=x)   +   \sum_{x} \E[S_2x]\, \pp(X=x)  \pp( S_1 x>\tau)\\
& = \sum_x \pp(X=x) \left( \E\left[S_1 x \wedge \tau \right] \,  +  \E[S_2 x]  \pp( S_1 x>\tau) \right).
\end{align*}
Since $\E[\eta_1]=\E[X]\E[S]$, we notice that
\begin{equation}
\label{condition}
\E[\eta_1] - \frac{\rho(\tau)}{\lambda} 
 =
\sum_x  \pp(X=x) \Big( \E[S x] \,\pp( S x\le \tau)- \E\left[S x\wedge \tau \right] \Big)
\end{equation}
and within Assumption~\ref{as_z} it is clear that the RHS of \eqref{condition} is strictly greater than zero when $\tau=z$, for some~$z\in\mathbb{R}_+$.
The final part of the proposition is straightforward.

\section*{Optimal Stopping Formulation of Theorem \ref{th_optimal_timeout}}

First, we explain how the proof of Theorem \ref{th_optimal_timeout} can be expressed as an optimal stopping problem; in the following, stopping time and timeout are used interchangeably. Then, we show how the optimal timeout can be verified as an application of the one-step-lookahead principle. We also give a proof in continuous time and with deterministic stopping times. The more general proof  (continuous-time stochastic timeouts) is more involved and, due to space constraints, is included in the appendix.

If we wish to minimize the time spent by a job in the processing phase, then we must minimize
\begin{align} \label{Vload}
  V & = \min_{\tau \in \mathcal T} \mathbb E \left[ \eta_1 \wedge \tau + \eta_2 \mathbb I_{\left\{ \tau \leq \eta_1\right\}}
\right] \\
& =
\min_{\tau \in \mathcal T} \mathbb E \left[ \int_0^\tau \bar F_1(t) dt + \bareta_2(\tau) \bar F_1(\tau)\right]
\end{align}
where $\mathcal T$ is the set of stopping times on $\mathbb R_+$, the functions $f_1(t)$ and $\bar F_1(t)$ are respectively the pdf and ccdf of $\eta_1$, and $\bareta_2(t) \bydef \mathbb E [ \eta_2 \mid \eta_1 \geq t]$. 
The minimization \eqref{Vload} is an optimal stopping problem with continuation cost $\bar F_1(t)$ and stopping cost $\bar \eta_2(t)\bar F_1(t)$. 
We note that the objective function above is equal to $\rho(\tau) / \lambda$ where $\rho(\tau)$ is the induced load on the speculative queueing network for timeout $\tau$. Thus, we aim at finding the stopping time that minimizes the load. 



\medskip
\noindent \textbf{Discrete Time Argument.} 
What follows is a brief informal argument for why the stopping condition \eqref{stop} is correct.  
If we discretize time as ${\mathcal T}_\Delta:=\{ 0, \Delta, 2\Delta,...\}$ and if we are only allowed to stop on this restricted set of times (rather than $\mathcal T$), then the optimization for $V$ above is a discrete time Markov decision process. The Bellman equation for this problem is
\begin{align*}
  V(t) = \min\big\{  p_\Delta(t) [ \Delta +   V(t+\Delta) ], \bareta_2 (t) 
\big\}\, ,
\end{align*}
where we define
$  p_\Delta (t) := \mathbb P ( \eta_1 \geq  t+ \Delta | \eta_1 \geq t ) = 1 - \frac{f_1(t) \Delta}{\bar F_1(t)} + o(\Delta) \, .$

The one-step-look-ahead is known to be optimal for a wide class of optimal stopping rules \cite[Section 4.4]{bertsekas2011dynamic}. Here, we stop (or timeout) if it is better to stop now than continue one time step and then stop. In our case, it is easy to see that this corresponds to the condition:
\begin{align*}
\bareta_2(t) \leq p_\Delta(t) [ \Delta +   \bareta_2(t+\Delta) ] \, .
\end{align*}
The RHS term can be simplified:
\begin{align*}
&  p_\Delta(t) [ \Delta +   \bareta_2(t+\Delta) ] \\
& = \left( 1 - \frac{f_1(t) \Delta}{\bar F_1(t)} \right) \left[  \Delta + \Delta \bareta_2'(t) +   \bareta_2(t)\right] + o(\Delta) \\
&
=\bareta_2(t) + \Delta \left[ 1+ \bareta'_2(t) - \frac{f_1(t) }{\bar F_1(t)} \bareta_2(t)   \right] + o(\Delta) \, .
\end{align*}
Observing the term in square brackets above, we see that  up to terms of order $o(\Delta)$ the one-step look-ahead rule gives the condition to stop at $\tau$ whenever
\begin{align*}
  \frac{1+ \bareta'_2(\tau)}{\bareta_2(\tau)} \geq \frac{f_1(\tau) }{\bar F_1( \tau)}  \,.
\end{align*}
This gives the stopping rule stated in \eqref{stop}. For one-step look-ahead to be optimal, we require that this set is closed, meaning that whenever the stopping condition is satisfied it remains satisfied for all future times. This is the motivation for Assumption~\ref{OTas}. 

 \medskip
\noindent \textbf{Continuous Time Argument.} If we restrict ourselves to deterministic stopping times, then it is clear that $\tau$ is the optimal stopping time: by conditioning on the value of $\eta_1$, we have that the objective \eqref{Vload} satisfies
\begin{align*}
\mathbb E \left[ \eta_1 \wedge \tau + \eta_2 \mathbb I_{\left\{ \tau \leq \eta_1\right\}}
\right] = 
\int_0^\tau \bar F_1(t) dt + \bareta_2(\tau) \bar F_1(\tau)
\end{align*}
and a stationary point of this optimization satisfies
\begin{align*}
0& =  \bar F_1(\tau) + \bareta_2'(\tau) \bar F_1(\tau) - \bareta_2(\tau) f_1(\tau)\\
& = \bar F_1(\tau) \bareta_2(\tau)   \left[ \frac{1+ \bareta'_2(\tau)}{\bareta_2(\tau) } - \frac{f_1(\tau) }{\bar F_1( \tau)} \right] \, .
\end{align*}
Thus, under Assumption \ref{OTas}, the optimal stopping time condition is
\begin{align*}
  \frac{1+ \bareta'_2(\tau)}{\bareta_2(\tau) } = \frac{f_1(\tau) }{\bar F_1( \tau)} \,.
\end{align*}

\subsection*{Proof of Theorem \ref{th_optimal_timeout}}

We put together the stochastic discrete time argument with the deterministic continuous time argument developed above to give a formal proof. Towards this purpose, we must formulate the optimal stopping problem as a free boundary problem and then solve it for the optimal rule. The theory of free boundary problems and optimal stopping is given in detail in~\cite{peskir2006optimal}.

A value function must satisfy the following free boundary problem
\begin{subequations}\label{FB}
\begin{align}
& \bareta_2(t)  \geq V(t)\, , \label{FB2} \\
& 0 = 1+ V'(t) - \frac{f_1(t)}{\bar F_1(t)} V(t) \quad \text{ on } \quad \{ V(t) < \bareta_2(t)\} \, .\label{FB3}
\end{align}	
\end{subequations} 
It is show in Section 2.2 of \cite{peskir2006optimal} that a solution of \eqref{FB} defines the value of a policy that stops on the set $S=\inf\{ t:  V(t) \geq \bareta_2(t) \}$. 
Assuming that $V(t) < \bareta_2(t)$, the o.d.e. \eqref{FB3} can be solved as follows
\begin{align*}
&  0 = 1 + V'(t) - \frac{f_1(t)}{\bar F_1(t)} V(t)  \\
&\implies f_1(t) V(t) - \bar F_1(t) V'(t) = \bar F_1(t) \\
&\implies \frac{d }{dt} (-\bar F_1(t) V(t) ) = \bar F_1(t) \\
& \implies \bar F_1(t) V(t) = - \int_0^t \bar F_1(s)ds + A \\
& \implies V(t) = \frac{A}{\bar F_1(t)} - \frac{1}{\bar F_1(t)}\int_0^t \bar F_1(s)ds
\end{align*}
where $A$ is a constant which we will specify shortly. 

We now investigate times $\tau \in [0,\infty]$ where $V(\tau) = \bareta_2(\tau)$.
Substituting the above expression for $V(t)$, we notice that
\begin{equation*}
  \bareta_2(\tau) = \frac{A}{\bar F_1(\tau)} - \frac{1}{\bar F_1(\tau)} \int_0^\tau \bar F_1(s) ds\\
  \implies A = \bareta_2(\tau) \bar F_1(\tau) + \int_0^\tau \bar F_1(s) ds \, .
\end{equation*}
Thus, we see that the solutions to the free boundary problem, for which there exists a time with $V(\tau)=\bareta_2(\tau)$, take the form
\begin{align}
 &  V_\tau(t) = \bareta_2(\tau) \frac{\bar F_1(\tau)}{\bar F_1(t)} + \int_t^\tau \bar F_1(s)ds \notag \\
& = \bareta_2(\tau) +  \frac{1}{\bar F_1(t)}\int_t^\tau \bareta_2(\tau) f_1(s) + \bar F_1(s) ds \,,\quad t\leq \tau \, . \label{Vtag}
\end{align}
Here we write $V_\tau(t)$ to make the dependence on $\tau$ explicit. We now require the minimal solution, setting $t=0$ and differentiating with respect to $\tau$ gives
\begin{align*}
  \partial_\tau V_\tau(0) & = \bareta_2'(\tau)\bar F_1(\tau) - \bareta_2(\tau)  f_1(\tau) + \bar F_1(\tau)\\
  & =  \bareta_2(\tau) \bar F_1(\tau) \left[ \frac{1+ \bareta_2'(\tau)}{\bareta_2(\tau)} - \frac{f_1(\tau)}{\bar F_1(\tau) } \right]  \, .
\end{align*}
The term $\bareta_2(t) \bar F_1(\tau)$ is positive while the term in square brackets is monotone increasing from Assumption \ref{OTas}. 
Thus, from the term in square brackets above, we see that the condition is that $\tau^*$ is the minimal value such that the term in square brackets is positive. That is
\begin{align*}
\tau^* = \min\left\{ \tau \geq 0 :  \frac{1+ \bareta_2'(\tau)}{\bareta_2(\tau) } \geq \frac{f_1(\tau)}{\bar F_1(\tau) } \right\} \,.
\end{align*}
Thus \eqref{Vtag} with $\tau = \tau^*$ characterizes the solution up until time $\tau^*$. We note that for time $t> \tau^*$ the solution must satisfy $V(t)=\bareta_2(t)$. This can be seen because of the following argument. If $V(t)=\bareta_2(t)$ did not hold for all $t\geq \tau^*$ then we can choose a time $t$ such that $V(t)=\bareta_2(t)$, but the value function is strictly smaller immediately after time $t$; however, then under condition \eqref{FB3} 
\[
V'(t) = 1- \frac{f_1(t)}{\bar F_1(t)} V(t) = 1- \frac{f_1(t)}{\bar F_1(t)} \bareta_2(t) > \bareta'_2(t)
\]
where the inequality above holds since 
\begin{align*}
  \frac{1+ \bareta_2'(t)}{\bareta_2(t)} > \frac{f_1(t)}{\bar F_1(t)}
\end{align*}
for $t> \tau^*$. Thus we have $V'(t)> \bareta'_2(t)$. So we see that this leads to a contradiction since we assumed the function $V(t)$ decreases below $\bareta_2(t)$ immediately after time $t$. This proves that $V(t)=\bareta_2(t)$ for all $t\geq \tau^*$.
 
From this we see that the minimal value function solving the free boundary problem is 
\begin{align*}
  V(t) 
=
\begin{cases}
	\bareta_2(\tau^*) +  \frac{1}{\bar F_1(t)}\int_t^{\tau^*} [\bareta_2(\tau^*) f_1(s) + \bar F_1(s) ] ds
	&  t\leq \tau^* \\
\bareta_2(t) & t > \tau^* 
\end{cases}
\end{align*}
where $\tau^*$ is given above, which implies that the optimal stopping set is $S:=\{ t : V(t) = \bareta_2(t)\} = \{ t : t\geq \tau^*\}$. Thus, it is optimal to stop at time $\tau^*$ as required.  

Finally, we note that if $\eta_2$ is independent of $\eta_1$ then $\bareta_2(t)$ is a constant, $\E[\eta_2]$, and the optimal stopping condition reduces to the condition \eqref{indyTime}.

\end{document}